\documentclass[letterpaper, 10 pt, conference]{ieeeconf}  

\IEEEoverridecommandlockouts                              
\newtheorem{theorem}{Theorem}
\newtheorem{definition}{Definition}
\newtheorem{corollary}{Corollary}
\newtheorem{remark}{Remark}
\newtheorem{lemma}{Lemma}
\newtheorem{fact}{Fact}

\usepackage{graphics} 
\usepackage{epsfig} 
\usepackage{mathptmx} 
\usepackage{times} 
\usepackage{amsmath} 
\usepackage{amssymb}  
\usepackage[ruled]{algorithm2e}
\usepackage[noend]{algpseudocode}

\makeatletter
\DeclareRobustCommand\widecheck[1]{{\mathpalette\@widecheck{#1}}}
\def\@widecheck#1#2{%
    \setbox\z@\hbox{\m@th$#1#2$}%
    \setbox\tw@\hbox{\m@th$#1%
       \widehat{%
          \vrule\@width\z@\@height\ht\z@
          \vrule\@height\z@\@width\wd\z@}$}%
    \dp\tw@-\ht\z@
    \@tempdima\ht\z@ \advance\@tempdima2\ht\tw@ \divide\@tempdima\thr@@
    \setbox\tw@\hbox{%
       \raise\@tempdima\hbox{\scalebox{1}[-1]{\lower\@tempdima\box
\tw@}}}%
    {\ooalign{\box\tw@ \cr \box\z@}}}
\makeatother

\title{\LARGE \bf
Satisfiability Bounds for $\omega$-regular Properties in Interval-valued Markov Chains
}

\author{Maxence Dutreix and Samuel Coogan
\thanks{Maxence Dutreix is with the School of Electrical and Computer Engineering, Georgia Institute of Technology, Atlanta, 30332, USA        
        {\tt\small maxdutreix@gatech.edu}}%
\thanks{Samuel Coogan is with the School of Electrical and Computer Engineering and the School of Civil and Environmental Engineering, Georgia Institute of Technology, Atlanta, 30332, USA
        {\tt\small sam.coogan@gatech.edu}}%
}

\begin{document}

\maketitle
\thispagestyle{empty}
\pagestyle{empty}

\begin{abstract}

We derive an algorithm to compute satisfiability bounds for arbitrary $\omega$-regular properties in an Interval-valued Markov Chain (IMC) interpreted in the adversarial sense. IMCs generalize regular Markov Chains by assigning a range of possible values to the transition probabilities between states. In particular, we expand the automata-based theory of $\omega$-regular property verification in Markov Chains to apply it to IMCs. Any $\omega$-regular property can be represented by a Deterministic Rabin Automata (DRA) with acceptance conditions expressed by Rabin pairs. Previous works on Markov Chains have shown that computing the probability of satisfying a given $\omega$-regular property reduces to a reachability problem in the product between the Markov Chain and the corresponding DRA. We similarly define the notion of a product between an IMC and a DRA. Then, we show that in a product IMC, there exists a particular assignment of the transition values that generates a largest set of non-accepting states. Subsequently, we prove that a lower bound is found by solving a reachability problem in that refined version of the original product IMC. We derive a similar approach for computing a satisfiability upper bound in a product IMC with one Rabin pair. For product IMCs with more than one Rabin pair, we establish that computing a satisfiability upper bound is equivalent to lower-bounding the satisfiability of the complement of the original property.  A search algorithm for finding the largest accepting and non-accepting sets of states in a product IMC is proposed. Finally, we demonstrate our findings in a case study.

\end{abstract}

\section{INTRODUCTION}

\textit{Markov Chains} have been extensively used as an intuitive yet powerful mathematical tool for modeling systems evolving through time in a stochastic fashion. They allow us to answer critical questions about the behavior of the underlying systems, often specified in terms of symbolic temporal logics, and derive appropriate control strategies \cite{ding2014optimal} \cite{fu2014probably}. As a superset of \textit{Linear Temporal Logic} (LTL), \textit{$\omega$-regular properties} are of particular interest to us due to their expressiveness. One can easily translate natural language inquiries such as \textit{``Will the system eventually reach a good state''} or \textit{``Will the system never reach a bad state and visit a good state infinitely often?"} into well-defined regular expressions. A method for computing the probability of fulfilling any $\omega$-regular property in Markov Chains is described in \cite{baier2008principles}. However, this derivation assumes that the probabilities of transition from state to state are known exactly.
  
Accessing the true probabilities of transitions might be impossible in practice and their values may only be approximated, e.g. from collected data. Furthermore, there has been a growing interest in abstractions of stochastic hybrid systems \cite{abate2011approximate} \cite{zamani2014symbolic}, and the discretization of a stochastic continuous state space might sometimes result in a finite abstraction where the transitions between states cannot be expressed as a single number \cite{dutreix2018}. To account for this, Markov Chains are augmented into \textit{Interval-valued Markov Chains (IMC)} where the probabilities of transition from state to state are given to lie within some interval \cite{kozine2002interval} \cite{vskulj2009discrete}. A direct consequence of this characteristic is that the probability of satisfying temporal properties in an IMC has to be formulated as an interval as well for all initial states.

Depending on the context in which they are utilized, IMCs give rise to two different semantic interpretations. One may view an IMC as an imperfect representation of a unique underlying Markov Chain whose transition bounds are not known exactly; this is called the \textit{Uncertain Markov Chain} (UMC) interpretation of IMCs. On the other hand, IMCs can be interpreted in an adversarial sense where a new probability distribution consistent with the transition bounds is non-deterministically selected each time a state is visited. In this case, we refer to an IMC as an \textit{Interval Markov Decision Process} (IMDP).

In \cite{chatterjee2008model}, the authors discuss the feasibility of the model-checking problem in both interpretations of IMCs and its computational complexity for $\omega$-regular properties. Nevertheless, efficient algorithms for computing satisfiability bounds are not provided.  Such bounds prove valuable in certain applications, such as the targeted state-space refinement of hybrid systems where IMCs naturally arise. 

Satisfiability bounds were calculated in \cite{lahijanian2015formal} for the Probabilistic Computation Tree Logic (PCTL) in IMDPs, but PCTL cannot express useful specifications such as liveness properties, i.e. the infinitely repeated occurence of an event. An automaton-based stochastic technique that asymptotically converges to lower and upper bounds for LTL formulas in UMCs was developed in \cite{benedikt2013ltl}. To the best of our knowledge, a deterministic algorithm capable of finding satisfiability bounds for arbitrary $\omega$-regular properties in IMDPs has not been presented in the literature and is the main contribution of this paper. 

Our objective is to extend the automaton-based procedure in \cite{baier2008principles} to accommodate IMCs interpreted as IMDPs. All $\omega$-regular properties can be converted into \textit{Deterministic Rabin Automata} (DRA) whose acceptance conditions are described by sets of states grouped in pairs called \textit{Rabin Pairs} \cite{baier2008principles}. Constructing the Cartesian product of a Markov Chain with a DRA enables to compute the probability that the stochastic evolution of the Markov Chain's state fulfills the property encoded in the DRA. In particular, it was shown that this probability is equal to that of reaching special sets of states called accepting \textit{Bottom Strongly Connected Components} (BSCC) in the product Markov Chain. Unfortunately, such a straightforward procedure does not work in general for IMCs. Although a similar definition of the Cartesian product between an IMC and a DRA can be established, we observe that the set of accepting BSCCs depends on the assumed transition values in the resulting product IMC. The structure of a product IMC is indeed specifically determined by transitions which can either create or eliminate a path between two states, i.e. transitions with a zero probability lower bound and a non-zero upper bound.

Nonetheless, we first show in this paper that a particular instantiation on the transition values yields a largest set of so-called \textit{non-accepting states}. Then, we show that computing a lower bound on the satisfiability of the property expressed by the DRA reduces to a reachability problem on the non-accepting states in the refined product IMC. If the underlying DRA only has one Rabin pair, we conversely prove that an upper bound is found by solving a reachability problem for a particular refinement of the product IMC that generates the most \textit{accepting states}. In the case where the DRA possesses more than one Rabin pair, we show that an upper bound is calculated by lower-bounding the satisfiability for the complement property of the DRA. Furthermore, we describe an efficient algorithm for finding the largest sets of non-accepting and accepting states in a product IMC along with their appropriate refinement.  Lastly, we illustrate our algorithm through the study of an agent moving in space according to an IMC. 

The paper is organized as follows: in Section I, we introduce important concepts and notations; then, in Section II, we rigorously formulate the problem to be solved; in Section III, we derive the main concepts used for bounding the satisfiability of $\omega$-regular properties in IMCs  and we present an algorithm for finding the largest sets of accepting and non-accepting states; finally, we present a case study demonstrating our findings in Section IV.

\section{PRELIMINARIES}

An \textit{Interval-Valued Markov Chain (IMC)} \cite{dutreix2018} is a 5-tuple $\mathcal{I} = (Q, \widecheck{T}, \widehat{T}, \Pi, L)$ where:
\begin{itemize}
\setlength{\itemsep}{0pt}
\item $Q$ is a finite set of states,
\item $\widecheck{T}: Q \times Q \rightarrow [0, 1] $ maps pairs of states to a lower transition bound so that $\widecheck{T}_{Q_{j} \rightarrow Q_{\ell}} := \widecheck{T}(Q_{j}, Q_{\ell})$ denotes the lower bound of the transition probability from state $Q_{j}$ to state $Q_{\ell}$, and 
\item $\widehat{T}: Q \times Q \rightarrow [0, 1] $ maps pairs of states to an upper transition bound so that $\widehat{T}_{Q_{j} \rightarrow Q_{\ell}} := \widehat{T}(Q_{j}, Q_{\ell})$ denotes the upper bound of the transition probability from state $Q_{j}$ to state $Q_{\ell}$,
\item $\Pi$ is a finite set of atomic propositions,
\item $L : Q \rightarrow 2^{\Pi}$ is a labeling function that assigns a subset of $\Pi$ to each state $Q$,
\end{itemize}
and $\widecheck{T}$ and $\widehat{T}$ satisfy $\widecheck{T}(Q_j,Q_\ell)\leq \widehat{T}(Q_j,Q_\ell)$ for all $Q_j,Q_\ell\in Q$ and 
\begin{equation}
\label{eq:36}
\sum_{Q_\ell\in Q} \widecheck{T}(Q_j,Q_\ell)\leq 1\leq \sum_{Q_\ell\in Q}\widehat{T}(Q_j,Q_\ell)  
\end{equation}
 for all $Q_j\in Q$.\\
 
A \textit{Markov Chain} $\mathcal{M} = (Q, T, \Pi, L)$ is similarly defined with the difference that the transition probability function $T: Q \times Q \rightarrow [0, 1] $ satisfies $0 \leq T(Q_j, Q_\ell) \leq 1$ for all $Q_j, Q_{\ell}\in Q$ and $\sum_{Q_{\ell} \in Q} T(Q_j, Q_\ell) = 1$ for all $Q_j\in Q$. Markov Chains evolve in discrete time; at each discrete time step, the Markov Chain transitions from its current state $Q_{i}$ to a state $Q_{j}$ according to the probability distribution set by $T$. For any sequence of states $\pi = q_0 q_1 q_2 \ldots$ in $\mathcal{M}$, with $q_{j} \in Q$, $q_0$ is called an \textit{initial state}. \\

A Markov Chain $\mathcal{M}$ is said to be \textit{induced} by IMC $\mathcal{I}$ if for all $Q_j,Q_\ell\in Q$,
\begin{align}
\widecheck{T}(Q_j,Q_\ell) \leq T(Q_j, Q_\ell) \leq \widehat{T}(Q_j,Q_\ell) \; \;.
\end{align}

An IMC $\mathcal{I}_{2}$ with transition functions $\widecheck{T}_{2}$ and $\widehat{T}_{2}$ is said to be \textit{induced} by IMC $\mathcal{I}_{1}$ with transition functions $\widecheck{T}_{1}$ and $\widehat{T}_{1}$ if both $\mathcal{I}_{1}$ and $\mathcal{I}_{2}$ have the same $Q$, $\Pi$ and $L$, and, for all $Q_j,Q_\ell\in Q$,
\begin{align}
\widecheck{T}_{1}(Q_j,Q_\ell) \leq \widecheck{T}_{2}(Q_j,Q_\ell) \leq \widehat{T}_{2}(Q_j,Q_\ell) \leq \widehat{T}_{1}(Q_j,Q_\ell) \; \;.
\end{align}
\noindent In this case, it follows that any Markov Chain induced by $\mathcal{I}_{2}$ is also induced by $\mathcal{I}_{1}$.\\

An IMC $\mathcal{I}$ is said to be interpreted as an \textit{Interval Markov Decision Process} (IMDP) if, at each time step $k$, the external environment non-deterministically chooses a Markov chain $\mathcal{M}_k$ induced by $\mathcal{I}$ and the next transition occurs according to $\mathcal{M}_k$. A mapping $\mathcal{\nu}$ from any finite path $\pi = q_0 q_1\ldots q_k$ in $\mathcal{I}$ to a Markov Chain $\mathcal{M}_k$ is called an \textit{adversary}. The set of all possible adversaries of $\mathcal{I}$ is denoted by $\mathcal{\nu}_{\mathcal{I}}$.\\

An IMC $\mathcal{I}$ is said to be interpreted as an \textit{Uncertain Markov Chain} (UMC) if the external environment non-deterministically chooses a single Markov chain $\mathcal{M}_0$ at $k = 0$ and the sequence of states $\pi = q_0 q_1 q_2 \ldots$ is determined by the transition probabilities in $\mathcal{M}_0$.\\

A \textit{Deterministic Rabin Automaton (DRA)} \cite{baier2008principles} is a 5-tuple $\mathcal{A} = (S, \Sigma, \delta, s_0, Acc)$ where:
\begin{itemize}
\item $S$ is a finite set of states,
\item $\Sigma$ is an alphabet,
\item $\delta : Q \times \Sigma \rightarrow S$ is a transition function
\item $s_0$ is an initial state
\item $Acc \subseteq 2^{S} \times 2^{S}$. An element $(E_{i}, F_{i}) \in Acc$, with $E_{i}, F_{i} \subset S$, is called a \textit{Rabin Pair}.
\end{itemize}

The probability of satisfying $\omega$-regular property $\phi$ starting from initial state $Q_i$ in IMC $\mathcal{I}$ under adversary $\mathcal{\nu}$ is denoted by $\mathcal{P}_{\mathcal{I}[\mathcal{\nu}]}(Q_i \models \phi)$.
The greatest lower bound and the least upper bound probabilities of satisfying property $\phi$ starting from initial state $Q_i$ in IMC $\mathcal{I}$ are denoted by $\widecheck{\mathcal{P}}_{\mathcal{I}}(Q_i \models \phi)$ and $\widehat{\mathcal{P}}_{\mathcal{I}}(Q_i \models \phi)$ respectively.\\

$\mathcal{P}_{\mathcal{M}}(Q_i\models \Diamond U)$ for $U \subseteq Q$ denotes the probability of eventually reaching set $U$ from initial state $Q_i$ in Markov Chain $\mathcal{M}$.

\section{PROBLEM FORMULATION}

Let $\mathcal{I}$ be an IMC interpreted as an IMDP with a set of possible adversaries $\mathcal{\nu}_{\mathcal{I}}$ and a set of atomic propositions $\Pi$, and let $\phi$ be an $\omega$-regular property over alphabet $\Pi$ (for formal definitions of $\omega$-regular properties and alphabet, see \cite{baier2008principles}).  Our goal is to find a systematic and efficient method for calculating $\widecheck{\mathcal{P}}_{\mathcal{I}}(Q_i \models \phi)$ and $\widehat{\mathcal{P}}_{\mathcal{I}}(Q_i \models \phi)$ where, for any adversary $\mathcal{\nu} \in \mathcal{\nu}_{\mathcal{I}}$,
\begin{align}
\widecheck{\mathcal{P}}_{\mathcal{I}}(Q_i \models \phi) \leq \mathcal{P}_{\mathcal{I}[\mathcal{\nu}]}(Q_i \models \phi) \leq \widehat{\mathcal{P}}_{\mathcal{I}}(Q_i \models \phi) \; .
\end{align}

Our approach extends the work in \cite{baier2008principles} for the verification of regular Markov chains against $\omega$-regular properties using automata-based methods. First, we generate a DRA $\mathcal{A}$ that recognizes the language induced by property $\phi$. Such a DRA always exists and creating it is a well studied problem. Several algorithms exist to accomplish this task efficiently for a large subset of $\omega$-regular expressions \cite{klein2006experiments} \cite{babiak2013effective}. Then, we construct the product $\mathcal{I} \otimes \mathcal{A}$, which is itself an IMC.\\

\begin{definition}
Let $\mathcal{I} = (Q, \widecheck{T}, \widehat{T}, \Pi, L)$ be an Interval-valued Markov Chain and $\mathcal{A} = (S, 2^{\Pi}, \delta, s_0, Acc)$ be a Deterministic Rabin Automaton. The \textit{product} $\mathcal{I} \otimes \mathcal{A} = (Q \times S, \widecheck{T'}, \widehat{T'}, Acc', L')$ is an Interval-valued Markov Chain where:
\begin{itemize}
\item $Q \times S$ is a set of states,\\
\item $\widecheck{T'}_{ \left<Q_{j},s\right> \rightarrow \left<Q_{\ell},s'\right>} = 
\begin{cases}
\widecheck{T'}_{Q_{j} \rightarrow Q_{\ell}}, \;\; \text{if} \;\; s' = \delta(s, L(Q_{\ell}))\\ \;\;\; \;\;\;\;\;\;0, \;\;\;\;\;\; \text{otherwise}
\end{cases}$\mbox{}\\\\
\item $\widehat{T'}_{ \left<Q_{j},s\right> \rightarrow \left<Q_{\ell},s'\right>} = 
\begin{cases}
\widehat{T'}_{Q_{j} \rightarrow Q_{\ell}}, \;\; \text{if} \;\; s' = \delta(s, L(Q_{\ell}))\\ \;\;\; \;\;\;\;\;\;0, \;\;\;\;\;\; \text{otherwise}
\end{cases}$\mbox{}\\
\item $Acc' = \lbrace E_{1}, E_{2}, \ldots, E_{k}, F_{1}, F_{2}, \ldots, F_{k} \rbrace$ is a set of atomic propositions, where $E_{i}$ and $F_{i}$ are the sets in the Rabin pairs of $Acc$,
\item $L': Q \times S \rightarrow 2^{Acc'}$ such that $H \in L'(\left<Q_{j},s\right>)$ if and only if $s \in H$, for all $H \in Acc' $ and for all $j$.
\end{itemize}
\end{definition}\mbox{}\\

A Markov Chain $\mathcal{M} \otimes \mathcal{A}$ induced by $\mathcal{I} \otimes \mathcal{A}$ is called a \textit{product Markov Chain}.\\

The probability of satisfying $\phi$ from initial state $Q_i$ in a Markov Chain equals that of reaching an accepting \textit{Bottom Strongly Connected Component} (BSCC) from initial state $\left<Q_i,s_0 \right>$ in the product Markov Chain with $\mathcal{A}$ \cite{baier2008principles}.\\

\begin{definition}
Given a Markov Chain $\mathcal{M}$ with states $Q$, a subset $B \subseteq Q$ is called a \textit{Bottom Strongly Connected Component} (BSCC) of $\mathcal{M}$ if it satisfies the following conditions:
\begin{itemize}
\item $B$ is strongly connected, that is, for each pair of states $(q,t)$ in $B$, there exists a path fragment $q_{0}q_{1}\ldots q_n$ such that $T(q_i,q_{i+1}) > 0$ for $i = 0,\; 1,\; \ldots, \; n-1$, and $q_i \in B$ for $0 \leq i \leq n$ with $q_0 = q$ and $q_n = t$,
\item no proper superset of $B$ is strongly connected,
\item $\forall s \in B$, $\Sigma_{t \in B}T(s,t) = 1$.\\
\end{itemize}
\end{definition}
\begin{definition}
A Bottom Strongly Connected Component $B$ of a product Markov Chain $\mathcal{M} \otimes \mathcal{A}$ is said to be \textit{accepting} if:
\begin{align}
\exists i: & \Bigg( \; \exists \left<Q_{j},s_{\ell} \right> \in B \; . \; F_{i} \in L'(\left<Q_{j},s_{\ell} \right>) \; \Bigg) \nonumber \\ 
& \wedge \Bigg( \;  \forall \left<Q_{j},s_{\ell} \right> \in B \; . \; E_{i} \not \in L'(\left<Q_{j},s_{\ell} \right>) \; \Bigg).
\end{align}
\end{definition}\mbox{}\\

In words, every state in a BSCC $B$ is reachable from any state in $B$, and every state in $B$ only transitions to another state in $B$. Moreover, $B$ is accepting when at least one of its states maps to the accepting set of a Rabin pair, while no state in $B$ maps to the non-accepting set of that same pair.\\

\begin{definition}
A state of $\mathcal{M} \otimes \mathcal{A}$ is $\textit{accepting}$ if it belongs to an accepting BSCC. The set of accepting states in $\mathcal{M} \otimes \mathcal{A}$  is denoted by $U^A_{\mathcal{M} \otimes \mathcal{A}}$; a state is \textit{non-accepting} if it belongs to a BSCC that is not accepting. The set of non-accepting states in $\mathcal{M} \otimes \mathcal{A}$ is denoted by $U^N_{\mathcal{M} \otimes \mathcal{A}}$. We omit the subscripts when they are obvious from the context.\\
\end{definition}

Note that each product Markov Chain $\mathcal{M} \otimes \mathcal{A}$ induced by $\mathcal{I} \otimes \mathcal{A}$ simulates the behavior of $\mathcal{I}$ under some adversary $\mathcal{\nu} \in \mathcal{\nu}_{\mathcal{I}}$. Indeed, for any two states $Q_j$ and $Q_{\ell}$ in $\mathcal{I}$ and some states $s, s', s''$ and $s'''$ in $\mathcal{A}$, we allow $T_{ \left<Q_{j},s\right> \rightarrow \left<Q_{\ell},s'\right>}$ and $T_{ \left<Q_{j},s''\right> \rightarrow \left<Q_{\ell},s'''\right>}$ to assume different values in $\mathcal{M} \otimes \mathcal{A}$, which means that the transition probability between $Q_j$ and $Q_{\ell}$ is permitted to change depending on the history of the path in $\mathcal{I}$.\\

\begin{fact}
\cite{baier2008principles} For any adversary $\mathcal{\nu} \in \mathcal{\nu}_{\mathcal{I}}$ in $\mathcal{I}$, it holds that $\mathcal{P}_{\mathcal{I}[\nu]}(Q_i \models \phi)$ =  $\mathcal{P}_{(\mathcal{M} \otimes \mathcal{A})_{\nu}}( \left<Q_i, s_{0} \right> \models \Diamond U^{A})$, where $(\mathcal{M} \otimes \mathcal{A})_{\nu}$ denotes the product Markov Chain induced by $\mathcal{I} \otimes \mathcal{A}$ corresponding to adversary $\nu$.\\
\end{fact}

\begin{figure}
\vspace{0.3cm}
\begin{center}
\hspace*{-1.1cm}
\includegraphics[scale=0.33]{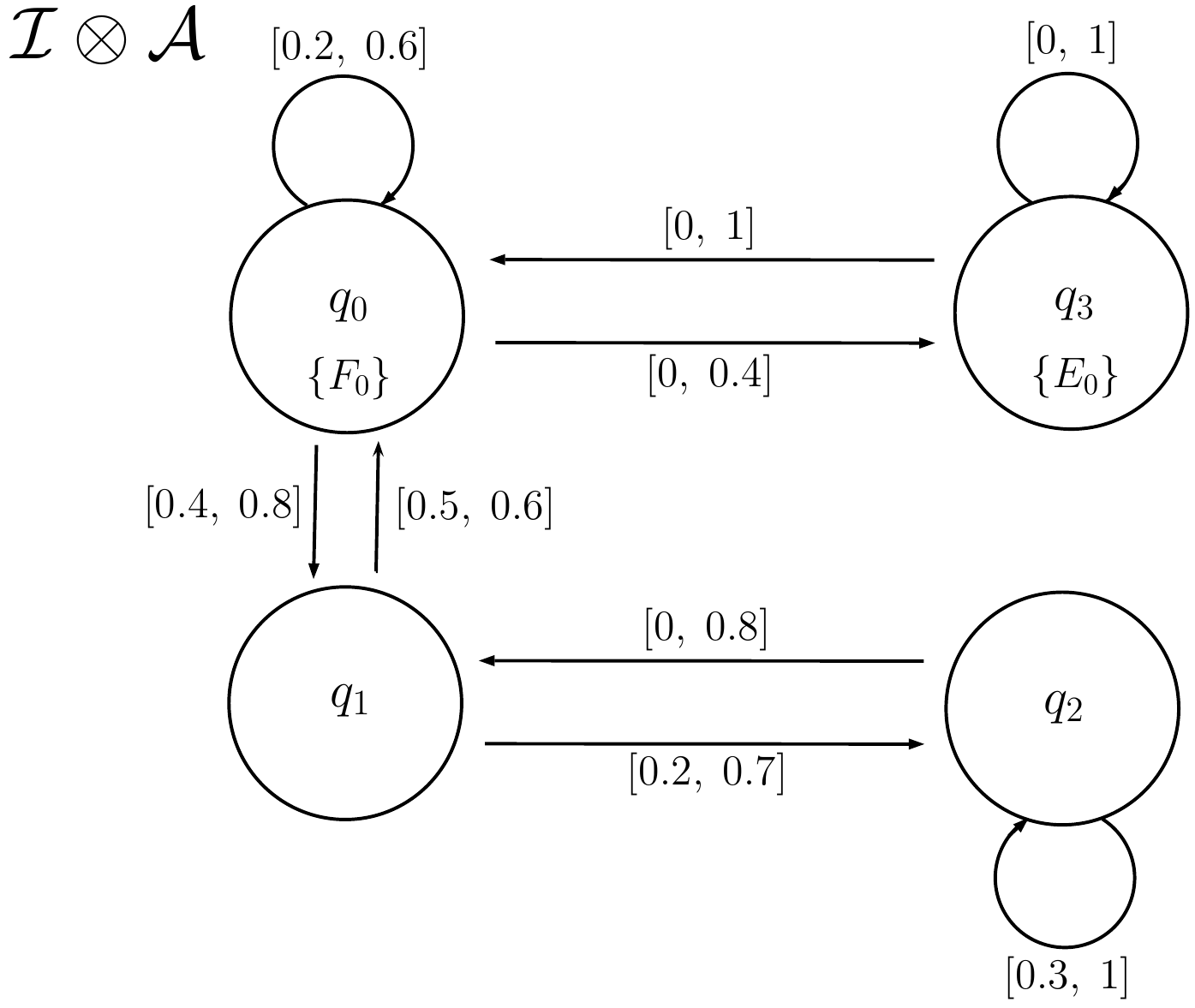}\\
\hspace*{-1.4cm}
\includegraphics[scale=0.33]{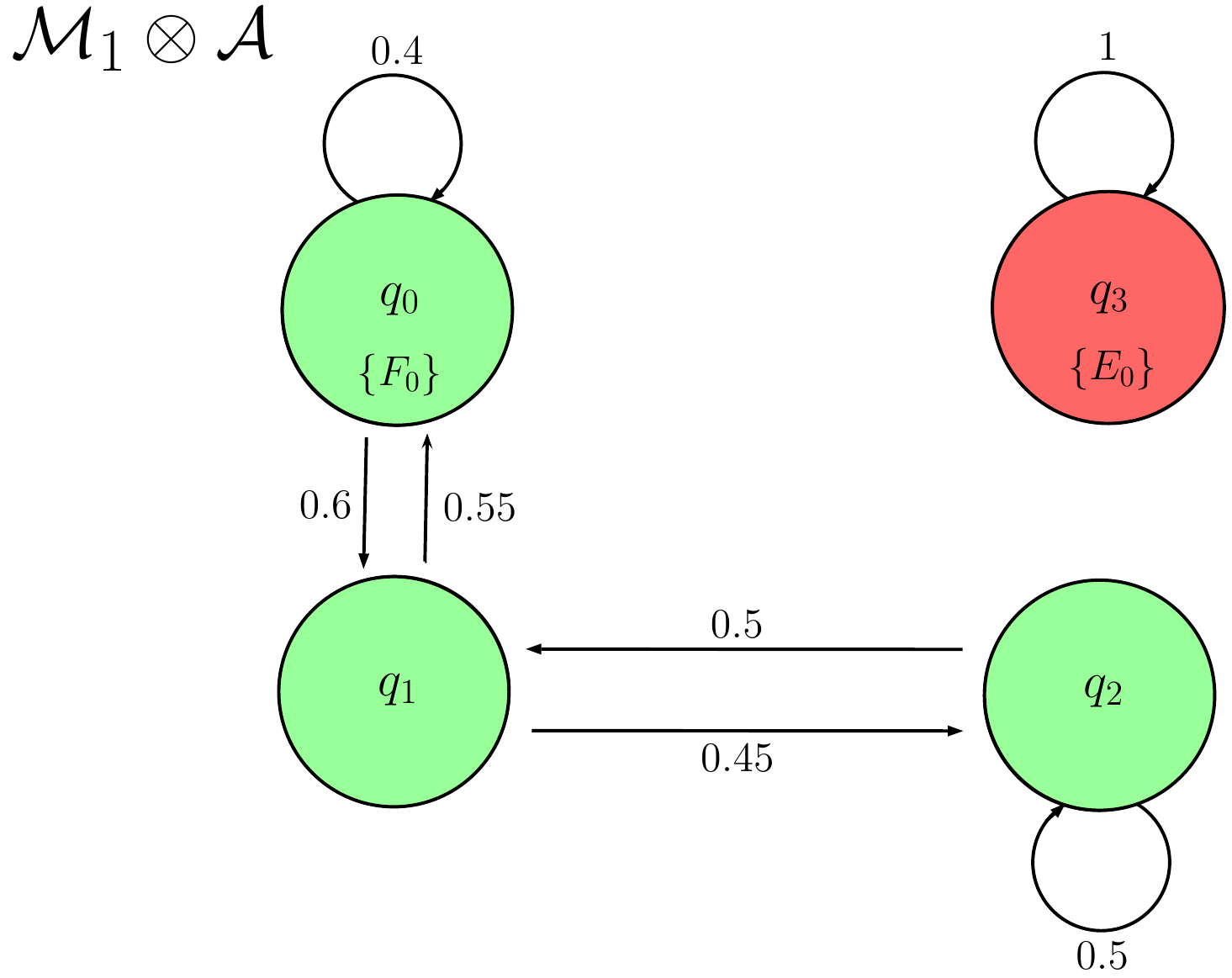}\\
\hspace*{-1.4cm}
\includegraphics[scale=0.33]{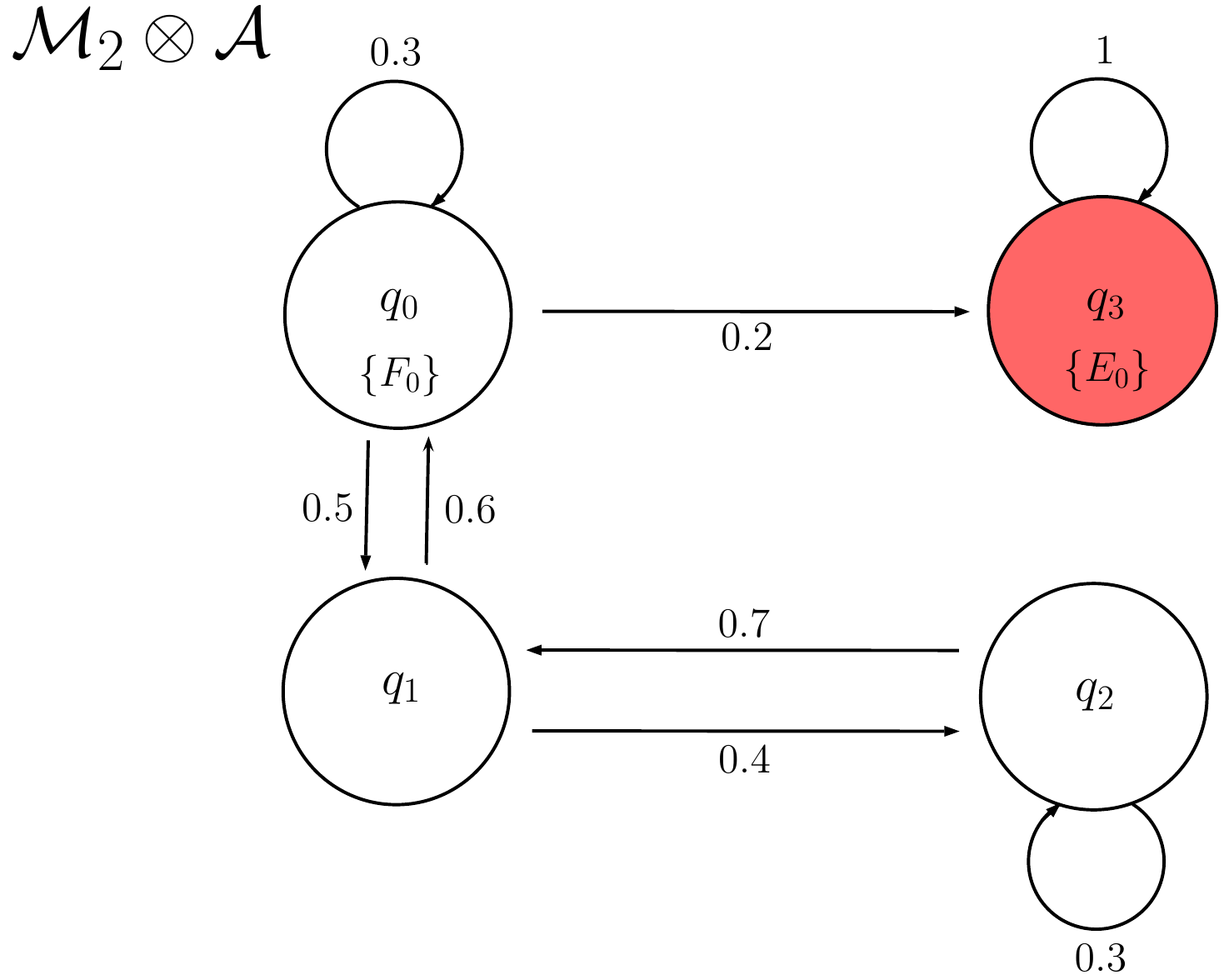}
\end{center}
\caption{Examples of two possible product Markov Chains $\mathcal{M}_{1} \otimes \mathcal{A}$ and $\mathcal{M}_2 \otimes \mathcal{A}$ induced by the product IMC $\mathcal{I} \otimes \mathcal{A}$. The sets of accepting states $U^{A}$ are shown in green whereas the sets of non-accepting states $U^{N}$ appear in red. In $\mathcal{M}_1 \otimes \mathcal{A}$, $U^{A} = \lbrace q_{0}, q_{1}, q_{2} \rbrace$ and $U^{N} = \lbrace q_{3} \rbrace$; in $\mathcal{M}_{2} \otimes \mathcal{A}$, $U^{A} = \lbrace \emptyset \rbrace$ and $U^{N} = \lbrace q_3\rbrace$.}
\end{figure}

It was shown in \cite{chen2013complexity} that the IMDP and UMC interpretations yield identical results for reachability problems. Consequently, computing $\widecheck{\mathcal{P}}_{\mathcal{I}}(Q_i \models \phi)$ and $\widehat{\mathcal{P}}_{\mathcal{I}}(Q_i \models \phi)$ amounts to finding the product Markov Chains induced by $\mathcal{I} \otimes \mathcal{A}$ that respectively minimize and maximize the probability of reaching an accepting state. Such reachability problems in IMCs have already been studied and solved when the destination states are fixed for all induced Markov Chains \cite{chatterjee2008model} \cite{lahijanian2015formal}. However, the set of accepting and non-accepting states may not be fixed in product IMCs and varies as a function of the assumed values for each transition. Specifically, $U^A$ and $U^N$ are determined by transitions that can be turned ``on" or ``off", i.e. those whose lower bound is zero and upper bound non-zero,  as seen in the example in Fig. 1. In this figure, in the product Markov Chain $\mathcal{M}_{1} \otimes \mathcal{A}$ induced by $\mathcal{I} \otimes \mathcal{A}$, the set of accepting states is $\lbrace q_{0}, q_{1}, q_{2} \rbrace$ while $\lbrace q_{3} \rbrace$ is non-accepting. However, in $\mathcal{M}_{2} \otimes \mathcal{A}$, the additional path from $q_0$ to $q_3$ prevents the existence of accepting states.

\textbf{Problem statement}: ``Given an IMC $\mathcal{I}$, an $\omega$-regular property $\phi$, and the DRA $\mathcal{A}$ corresponding to $\phi$, find the greatest lower bound and the least upper bound on the probability of reaching an accepting state from any initial state $\left<Q_i, s_{0} \right>$ in the product IMC $\mathcal{I} \otimes \mathcal{A}$, and thereby find the greatest lower bound and least upper bound on the probability of satisfying $\phi$ for any adversary $\mathcal{\nu}$ in $\mathcal{I}$ and for any initial state $Q_{i}$."

We emphasize that this problem is non-trivial due the dependence of the set of accepting states on the assumed values for the transitions whose lower bound is zero and whose upper bound is non-zero.
\section{BOUNDING THE SATISFIABILITY OF $\omega$-REGULAR PROPERTIES IN AN IMC}

In \cite{chatterjee2008model}, the authors discussed an algorithm for computing the probability bounds of reaching any fixed set of states in an IMC. We remarked in the previous section that, in general, the set of accepting and non-accepting states in a product IMC may depend on the assumed transition values. This is however not always the case. Let us define special classes of product IMCs.\\
\begin{definition}
A product IMC $\mathcal{I} \otimes \mathcal{A}$ is an \textit{Accepting-Static Interval-Valued Markov Chain (ASIMC)} if for any two product Markov Chains $\mathcal{M}_{1} \otimes \mathcal{A}$ and $\mathcal{M}_{2} \otimes \mathcal{A}$ induced by $\mathcal{I} \otimes \mathcal{A}$, it holds that
$(U^A)_{\mathcal{M}_{1} \otimes \mathcal{A}} = (U^A)_{\mathcal{M}_{2} \otimes \mathcal{A}}$.\\
\end{definition}
\begin{definition}
A product IMC $\mathcal{I} \otimes \mathcal{A}$ is an \textit{Non-Accepting-Static Interval-Valued Markov Chain (NASIMC)} if for any two product Markov Chains $\mathcal{M}_{1} \otimes \mathcal{A}$ and $\mathcal{M}_{2} \otimes \mathcal{A}$ induced by $\mathcal{I} \otimes \mathcal{A}$, it holds that
$(U^N)_{\mathcal{M}_{1} \otimes \mathcal{A}} = (U^N)_{\mathcal{M}_{2} \otimes \mathcal{A}}$.\\
\end{definition}
\noindent In ASIMCs, the set of accepting states remains the same for all induced product Markov Chains, while this property holds true for non-accepting states in NASIMCs. Therefore, we can apply the standard reachability techniques in \cite{chatterjee2008model} to compute bounds on $\mathcal{P}_{\mathcal{I} \otimes \mathcal{A}}(\left<Q_i,s_0\right> \models \Diamond U^A)$ or $\mathcal{P}_{\mathcal{I} \otimes \mathcal{A}}(\left<Q_i,s_0\right> \models \Diamond U^N)$ in these particular classes of product IMCs.

Notice that any product IMC $\mathcal{I} \otimes \mathcal{A}$ induces at most  $(|Q| \cdot |S|)^{|Q| \cdot |S|}$ combinations of ``on" and ``off" transitions. Thus, a computationally inefficient solution to our problem is to consider every such combinations. Then, for all resulting product IMCs, we bound the probability of reaching the induced $U^A$ from the initial state of interest, and retain the smallest lower bound and greatest upper bound obtained as absolute bounds for the satisfiability of $\phi$. 

In this section, we develop a more efficient method for computing satisfiability bounds for a given $\omega$-regular property $\phi$. We consider two different approaches for finding a lower bound and an upper bound due to the acceptance condition of Rabin automata. Specifically, we prove that all product IMCs induce a worst-case NASIMCs containing the largest set of non-accepting states and in which the probability of reaching an accepting BSCC is minimized from any initial state. Then, we show that the converse best-case ASIMCs are always induced by product IMCs with one Rabin pair, and the probability of reaching an accepting BSCC is maximized from any initial state in those ASIMCs. If the automata corresponding to $\phi$ possesses more than one Rabin pair, we determine an upper bound by computing a lower bound on the satisfiability for the complement of $\phi$. Finally, we suggest a search algorithm for efficiently finding the largest sets of accepting and non-accepting states.

\subsection{Lower Bound Computation}

A key observation is that any infinite sequence of states in a Markov Chain eventually reaches a BSCC.\\

\begin{lemma}
\cite{baier2008principles} For any infinite sequence of states $\pi = q_{0}q_{1}q_{2}\ldots$ in a Markov Chain, there exists an index $i \geq 0$ such that $q_i$ belongs to a BSCC.\\
\end{lemma}
The following corollary relies on the fact that a BSCC is either accepting or non-accepting.\\

\begin{corollary}
For any initial state $ \left< Q_i, s_0 \right>$ in a product Markov Chain $\mathcal{M} \otimes {\mathcal{A}}$,
\begin{align}
\mathcal{P}_{\mathcal{M} \otimes \mathcal{A}}( \left<Q_i,s_0 \right> \models \Diamond U^A) = 1- \mathcal{P}_{\mathcal{M} \otimes \mathcal{A}}(\left< Q_i,s_0 \right> \models \Diamond U^N) \; \; .
\end{align}

\end{corollary}

\begin{proof}
We denote the union of all BSCCs in $\mathcal{M} \otimes {\mathcal{A}}$ by $BSCC(\mathcal{M} \otimes {\mathcal{A}})$. From Lemma 1, it follows that 
\begin{align*}
& \mathcal{P}_{\mathcal{M} \otimes \mathcal{A}}\big( \left<Q_i,s_0 \right> \models \Diamond BSCC(\mathcal{M} \otimes {\mathcal{A}}) \;\big) \\ & = \mathcal{P}_{\mathcal{M} \otimes \mathcal{A}}( \left<Q_i,s_0 \right> \models \Diamond U^A) + \mathcal{P}_{\mathcal{M} \otimes \mathcal{A}}( \left< Q_i,s_0\right> \models \Diamond U^N) = 1.
\end{align*}
\end{proof}\mbox{}\\
\begin{lemma}
Let $\mathcal{I} \otimes  \mathcal{A}$ be a product IMC. Let $(\mathcal{I} \otimes  \mathcal{A})_1$ and $(\mathcal{I} \otimes  \mathcal{A})_2$ be two product NASIMCs induced by $\mathcal{I} \otimes  \mathcal{A}$ with sets of non-accepting states $U^N_1$ and $U^N_2$ respectively. There exists an NASIMC $(\mathcal{I} \otimes  \mathcal{A})_3$ induced by $\mathcal{I} \otimes  \mathcal{A}$ with non-accepting states $U^N_3$ and such that $ \big(U^N_1 \cup U^N_2) \subseteq U^N_3$.\\
\end{lemma}
\begin{proof}
This proof is constructive. Let $\widecheck{T}$, $\widehat{T}$,  $\widecheck{T}_{1}$, $\widehat{T}_{1}$,  $\widecheck{T}_{2}$, $\widehat{T}_{2}$ and  $\widecheck{T}_{3}$, $\widehat{T}_{3}$ be the transition bounds functions in the product IMCs  $\mathcal{I} \otimes  \mathcal{A}$, $(\mathcal{I} \otimes  \mathcal{A})_{1}$, $(\mathcal{I} \otimes  \mathcal{A})_{2}$ and $(\mathcal{I} \otimes  \mathcal{A})_{3}$ respectively. \\
\begin{itemize}
\item Case $U^N_1 \cap U^N_2 = \emptyset \;$: \\\\ Set $\widecheck{T}_{3}(Q_{i}, Q_{j}) = \widecheck{T}_{1}(Q_{i}, Q_{j})$ and $\widehat{T}_{3}(Q_{i}, Q_{j}) = \widehat{T}_{1}(Q_{i}, Q_{j})$  for all $Q_{i} \in U^N_1$ and for all $Q_{j} \in Q \times S$. Set $\widecheck{T}_{3}(Q_{i}, Q_{j}) = \widecheck{T}_{2}(Q_{i}, Q_{j})$ and $\widehat{T}_{3}(Q_{i}, Q_{j}) = \widehat{T}_{2}(Q_{i}, Q_{j})$  for all $Q_{i} \in U^N_2$ and for all $Q_{j} \in Q \times S$. We can do this because $U^N_1$ and $U^N_2$ are disjoint and the transitions leaving from any state in $U^N_1$ are independent from the transitions leaving from any state in $U^N_2$. Finally, set $\widecheck{T}_{3}(Q_{i}, Q_{j}) = \widecheck{T}(Q_{i}, Q_{j})$ and $\widehat{T}_{3}(Q_{i}, Q_{j}) = \widehat{T}(Q_{i}, Q_{j})$ for all other transitions. The product IMC $(\mathcal{I} \otimes  \mathcal{A})_{3}$ always induces product Markov Chains with sets of non-accepting states containing $U^N_1 \cup U^N_2$.\\
\end{itemize}

\begin{itemize}
\item Case $U^N_1 \cap U^N_2 \not = \emptyset \;$: \\\\ Let $D = U^N_1 \cap U^N_2$. Set $\widecheck{T}_{3}(Q_{i}, Q_{j}) = \widecheck{T}_{1}(Q_{i}, Q_{j})$ and $\widehat{T}_{3}(Q_{i}, Q_{j}) = \widehat{T}_{1}(Q_{i}, Q_{j})$  for all $Q_{i} \in U^N_1 \smallsetminus D$ and for all $Q_{j} \in Q \times S$. Set $\widecheck{T}_{3}(Q_{i}, Q_{j}) = \widecheck{T}_{2}(Q_{i}, Q_{j})$ and $\widehat{T}_{3}(Q_{i}, Q_{j}) = \widehat{T}_{2}(Q_{i}, Q_{j})$  for all $Q_{i} \in U^N_2 \smallsetminus D$ and for all $Q_{j} \in Q \times S$.  Set $\widecheck{T}_{3}(Q_{i}, Q_{j}) > 0$ for all $Q_{i}, Q_{j} \in D$ such that $\widecheck{T}_{1}(Q_{i}, Q_{j}) > 0$ and $\widecheck{T}_{2}(Q_{i}, Q_{j}) > 0$. Set $\widehat{T}_{3}(Q_{i}, Q_{j}) = 0$ for all $Q_{i} \in D$  and $Q_{j} \not \in U^N_1 \cup U^N_2$. $U^N_1 \cup U^N_2$ is now a BSCC in $(\mathcal{I} \otimes \mathcal{A})_{3}$. In addition, for any state in $U^N_1$ that maps to an accepting set $F_{i}$, there has to be a state in $U^N_1$ that maps to the corresponding non-accepting set $E_{i}$ since $U^N_1$ was made non-accepting in $(\mathcal{I} \otimes  \mathcal{A})_{1}$. The same reasoning holds for $U^N_2$. Therefore, $U^N_1 \cup U^N_2$ is non-accepting and the resulting product IMC always induces product Markov Chains with sets of non-accepting states containing $U^N_1 \cup U^N_2$.
\end{itemize}
\end{proof}

Lemma 2 implies the existence of an ASIMC whose set of non-accepting states is the ``largest", in the sense that it contains all sets of non-accepting state which can be induced by a product IMC. \\

\begin{corollary}
Let $\mathcal{I} \otimes \mathcal{A}$ be a product IMC. There exists a NASIMC  induced by $\mathcal{I} \otimes  \mathcal{A}$ with a set of non-accepting states $U^N_{\ell}$ and such that $U^N_{i} \subseteq U^N_{\ell}$, where $U^N_{i}$ is the set of non-accepting states for any product Markov Chain $(\mathcal{M} \otimes \mathcal{A})_{i}$ induced by $\mathcal{I} \otimes \mathcal{A}$.\\
\end{corollary}
\begin{remark} 
Let $[\mathcal{I} \otimes \mathcal{A}]^{N}$ be the set of all NASIMCs induced by $\mathcal{I} \otimes \mathcal{A}$ producing non-accepting states $U^N_{\ell}$ from Corollary 2. We denote the transition bounds functions of $\mathcal{I} \otimes \mathcal{A}$ by $\widecheck{T}$ and $\widehat{T}$. There exists a non-empty set of NASIMCs $[\mathcal{I} \otimes \mathcal{A}]^{N}_{\ell} \subseteq [\mathcal{I} \otimes \mathcal{A}]^{N}$ such that, for all $(\mathcal{I} \otimes \mathcal{A})_{\ell} \in [\mathcal{I} \otimes \mathcal{A}]^{N}_{\ell}$ with transition functions $\widecheck{T}_{\ell}$ and $\widehat{T}_{\ell}$, $\widecheck{T}_{\ell}(  \;\left< Q_{i}, s_i \right>,  \left< Q_{j}, s_j \right> \; ) = \widecheck{T}( \; \left< Q_{i}, s_i \right>,  \left< Q_{j}, s_j \right> \; )$ and  $\widehat{T}_{\ell}( \; \left< Q_{i}, s_i \right> ,\left< Q_{j}, s_j \right> \;) = \widehat{T}( \; \left< Q_{i}, s_i \right>,  \left< Q_{j}, s_j \right> \;)$ for all $\left< Q_{i}, s_i \right> \not \in  U^N_{\ell}$ and all $\left< Q_{j}, s_j \right> \in  Q \times S$. \\
\end{remark}

\noindent Remark 1 is due to the fact that the transition intervals from the states outside of $U^N_{\ell}$ have no effect on $U^N_{\ell}$ being non-accepting. Therefore, for any NASIMC producing $U^N_{\ell}$, we can set the interval of these transitions to the ones in the original product IMC without affecting the set of non-accepting states.\\

\indent Now, consider two sets of non-accepting states $U^N_1$ and $U^N_2$ which can possibly be induced by $\mathcal{I} \otimes \mathcal{A}$, with $U^N_2$ being fully contained in $U^N_1$. The next step consists in proving that if an NASIMC induced by $\mathcal{I} \otimes \mathcal{A}$ has a set of non-accepting states $U^N_2$, there exists a NASIMC induced by $\mathcal{I} \otimes \mathcal{A}$ with set of non-accepting states $U^N_1$ such that the upper bound probability of reaching $U^N_2$ in the first NASIMC is less than or equal to that of reaching $U^N_1$ in the second NASIMC for all initial states.\\

\begin{lemma}
Let $\mathcal{I} \otimes  \mathcal{A}$ be a product IMC. Consider two sets of non-accepting states $U^N_1$ and $U^N_2$ which can be induced by $\mathcal{I} \otimes  \mathcal{A}$ and such that $U^N_2 \subseteq U^N_1$. For any NASIMC $(\mathcal{I} \otimes \mathcal{A})_2$ with non-accepting states $U^N_2$ induced by $\mathcal{I} \otimes  \mathcal{A}$, there exists a NASIMC $(\mathcal{I} \otimes \mathcal{A})_1$ with non-accepting states $U^N_1$ induced by $\mathcal{I} \otimes  \mathcal{A}$ such that, for any initial state $\left<Q_i, s_0\right>$, 
\begin{align}
\mathcal{\widehat{P}}_{(\mathcal{I} \otimes \mathcal{A})_1}(\left<Q_i,s_0\right> \models \Diamond U^N_1) \geq \mathcal{\widehat{P}}_{(\mathcal{I} \otimes \mathcal{A})_2}(\left<Q_i,s_0\right> \models \Diamond U^N_2)  \; \; .
\end{align}

\end{lemma}

\begin{proof}
We provide a proof sketch due to space constraint. Let $(\mathcal{I} \otimes \mathcal{A})_2$ be a NASIMC induced by $\mathcal{I} \otimes  \mathcal{A}$ with non-accepting states $U^N_2$. Construct an IMC $(\mathcal{I} \otimes \mathcal{A})_1$ where the transitions from the states in $U^N_1$ have the same probability intervals as in $\mathcal{I} \otimes  \mathcal{A}$, and all others transitions are the same as in $(\mathcal{I} \otimes \mathcal{A})_2$. Then, in $(\mathcal{I} \otimes \mathcal{A})_1$, fix all values of the transitions from the states in $U^N_1$ such that $U^N_1$ is rendered non-accepting --- such a combination exists by assumption. $(\mathcal{I} \otimes \mathcal{A})_1$ has to be a NASIMC with non-accepting states $U^N_1$.\\
\begin{itemize}
\item If $\left<Q_i,s_0\right> \in U^N_1$, $\mathcal{\widehat{P}}_{(\mathcal{I} \otimes \mathcal{A})_1}(\left<Q_i,s_0\right> \models \Diamond U^N_1) = 1 \\
\geq \mathcal{\widehat{P}}_{(\mathcal{I} \otimes \mathcal{A})_2}(\left<Q_i,s_0\right> \models \Diamond U^N_2)$,\\
\item By construction, the transition intervals from the states not in $U^N_1$ to all other states are the same in both $(\mathcal{I} \otimes \mathcal{A})_1$ and $(\mathcal{I} \otimes \mathcal{A})_2$. Since $U^N_2 \subseteq U^N_1$, if $\left<Q_i,s_0\right> \not \in U^N_1$, it must be true that  $\mathcal{\widehat{P}}_{(\mathcal{I} \otimes \mathcal{A})_1}(\left<Q_i,s_0\right> \models \Diamond U^N_1) \geq \mathcal{\widehat{P}}_{(\mathcal{I} \otimes \mathcal{A})_2}(\left<Q_i,s_0\right> \models \Diamond U^N_2)$
\end{itemize}
\end{proof}\mbox{}\\

\begin{lemma}
Let $\mathcal{I} \otimes  \mathcal{A}$ be a product IMC. Let $[\mathcal{I} \otimes \mathcal{A}]^{N}$ and $[\mathcal{I} \otimes \mathcal{A}]^{N}_{\ell}$ be the sets as defined in Remark 1. For any NASIMC $ (\mathcal{I} \otimes \mathcal{A})^{'} \in [\mathcal{I} \otimes \mathcal{A}]^{N}$, any NASIMC $ (\mathcal{I} \otimes \mathcal{A})_{\ell} \in [\mathcal{I} \otimes \mathcal{A}]^{N}_{\ell}$ and any initial state $\left<Q_i, s_0\right>$,
\begin{align*}
\widehat{\mathcal{P}}_{(\mathcal{I} \otimes \mathcal{A})_{\ell}}(\left<Q_i, s_0\right> \models \Diamond U^N_{\ell}) \geq \widehat{\mathcal{P}}_{(\mathcal{I} \otimes \mathcal{A})^{'}}(\left<Q_i, s_0\right> \models \Diamond U^N_{\ell})\;\;. \\
\end{align*}
\end{lemma}
\begin{proof}
If $\left<Q_i, s_0\right> \in U^{N}_{\ell}$, then $\widehat{\mathcal{P}}_{(\mathcal{I} \otimes \mathcal{A})_{\ell}}(\left<Q_i, s_0\right> \models \Diamond U^N_{\ell})  = \widehat{\mathcal{P}}_{(\mathcal{I} \otimes \mathcal{A})^{'}}(\left<Q_i, s_0\right> \models \Diamond U^N_{\ell}) = 1$. If $\left<Q_i, s_0\right> \not \in U^{N}_{\ell}$, the inequality follows from the fact that, by the definition of $[\mathcal{I} \otimes \mathcal{A}]^{N}_{l}$, for any Markov Chain induced by $(\mathcal{I} \otimes \mathcal{A})^{'}$, there exists a Markov Chain induced by $(\mathcal{I} \otimes \mathcal{A})_{\ell}$ with similar transition values from all states not in $U^{N}_{\ell}$.
\end{proof}\mbox{}\\

We call $[\mathcal{I} \otimes \mathcal{A}]^{N}_{\ell}$ the set of the worst case NASIMCs. 
Theorem 1 claims that the computation of a lower bound on $\phi$ from any initial state reduces to a reachability problem in any NASIMC in the set $[\mathcal{I} \otimes \mathcal{A}]^{N}_{\ell}$.\\

\begin{theorem}
Let $\mathcal{I}$ be an IMC and $\mathcal{A}$ be a DRA corresponding to the $\omega$-regular property $\phi$. Let $(\mathcal{I} \otimes \mathcal{A})_{\ell}$ and $(\mathcal{I} \otimes \mathcal{A})^{'}$  be any two NASIMCs from the set $[\mathcal{I} \otimes \mathcal{A}]^{N}_{\ell}$ defined in Remark 1. It holds that $\mathcal{\widehat{P}}_{(\mathcal{I} \otimes \mathcal{A})_{\ell}}(\left<Q_i, s_0\right> \models \Diamond U^N_{\ell}) = \mathcal{\widehat{P}}_{(\mathcal{I} \otimes \mathcal{A})^{'}}(\left<Q_i, s_0\right> \models \Diamond U^N_{\ell})$ and, for any state $Q_i \in \mathcal{I}$,

\begin{align}
\widecheck{\mathcal{P}}_{\mathcal{I}}(Q_i  \models \phi) = 1 - \mathcal{\widehat{P}}_{(\mathcal{I} \otimes \mathcal{A})_{\ell}}(\left<Q_i, s_0\right> \models \Diamond U^N_{\ell}) \; \;.
\end{align}

\end{theorem}\mbox{}\\

\begin{proof}
Corollary 1 implies that, for all $(\mathcal{M} \otimes \mathcal{A})_{\nu}$ induced by $\mathcal{I} \otimes \mathcal{A}$ corresponding to some adversary $\nu \in \nu_{\mathcal{I}}$ of $\mathcal{I}$, $\mathcal{P}_{(\mathcal{M} \otimes \mathcal{A})_{\nu}}(\left<Q_i, s_0\right> \models \Diamond U^A) = \mathcal{P}_{\mathcal{I}[\nu]}(Q_i  \models \phi)$ is minimized when $\mathcal{P}_{(\mathcal{M} \otimes \mathcal{A})_{\nu}}(\left<Q_i, s_0\right> \models \Diamond U^N)$ is maximized. By Lemma 3 and Lemma 4, the maximum value for $\mathcal{P}_{(\mathcal{M} \otimes \mathcal{A})_{\nu}}(\left<Q_i, s_0\right> \models \Diamond U^N)$ is reached for some $(\mathcal{M} \otimes \mathcal{A})_{\nu}$ induced by the NASIMC $(\mathcal{I} \otimes \mathcal{A})_{\ell}$.
\end{proof}

\subsection{Upper Bound Computation}

One could think of a similar approach here and compute a satisfiability upper bound by maximizing the probability of transition to an accepting BSCC. However, due to the acceptance condition of Rabin Automata, the analogous version of Lemma 2 for accepting states does not always hold true, as shown in Fig. 2. We consequently treat two different types of product IMCs separately: those endowed with only one Rabin pair --- that is, $|Acc| = 1$ in $\mathcal{A}$ --- and those with more than one Rabin pair.\\

\begin{figure}
\vspace{0.2cm}
\begin{center}
\includegraphics[scale=0.6]{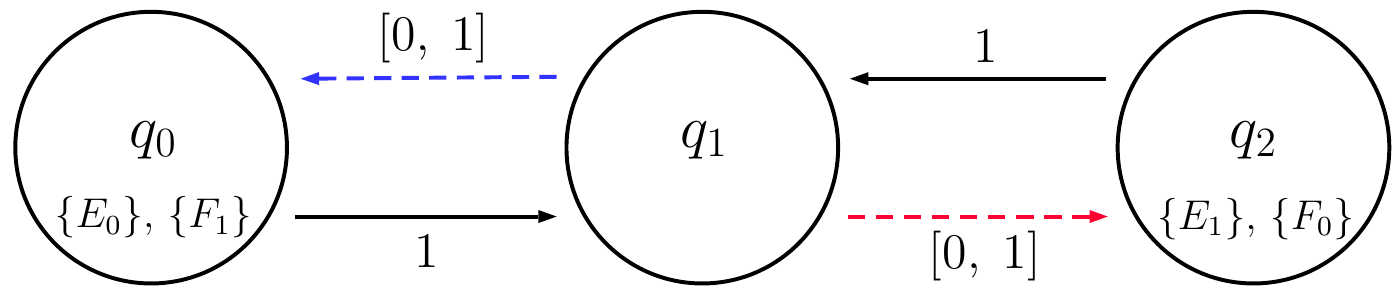}
\end{center}
\caption{Example where the analogous version of Lemma 2 for accepting states does not hold. Setting the blue transition to 1 and the red transition to 0 makes $q_0$ and $q_1$ accepting. Conversely, setting the blue transition to 1 and the red transition to 0 renders $q_1$ and $q_2$ accepting. Nonetheless, no assignment makes $q_0$, $q_1$ and $q_2$ accepting at the same time.}
\end{figure}

\subsubsection{Product IMC with one Rabin pair}\mbox{}\\

The theorem and lemma in this section are similar to the ones in section IV  A, and are provided without proof due to space constraints.\\

We denote by $U^A_{u}$ the largest set of accepting states induced by $\mathcal{I} \otimes \mathcal{A}$. We define the set of best case ASIMCs $[\mathcal{I} \otimes \mathcal{A}]^{A}_{u}$ analogously to the set $[\mathcal{I} \otimes \mathcal{A}]^{N}_{l}$ for NASIMCs.\\

\begin{theorem}
Let $\mathcal{I}$ be an IMC and $\mathcal{A}$ be a DRA with one Rabin pair corresponding to the $\omega$-regular property $\phi$. Let $(\mathcal{I} \otimes \mathcal{A})_{u}$ and $(\mathcal{I} \otimes \mathcal{A})^{'}$ be any two ASIMCs from the set $[\mathcal{I} \otimes \mathcal{A}]^{A}_{u}$. It holds that $\mathcal{\widehat{P}}_{(\mathcal{I} \otimes \mathcal{A})_{u}}(\left<Q_i, s_0\right> \models \Diamond U^A_{u}) = \mathcal{\widehat{P}}_{(\mathcal{I} \otimes \mathcal{A})^{'}}(\left<Q_i, s_0\right> \models \Diamond U^A_{u})$ and,  for any state $Q_i \in \mathcal{I}$,
\begin{align}
\mathcal{\widehat{P}}_{\mathcal{I}}(Q_i  \models \phi) =  \mathcal{\widehat{P}}_{(\mathcal{I} \otimes \mathcal{A})_{u}}(\left<Q_i, s_0\right> \models \Diamond U^A_{u}) \;\; .
\end{align}
\end{theorem}\mbox{}\\
\subsubsection{Product IMC with more than one Rabin pairs}\mbox{}\\

We previously observed that product IMCs with more than one Rabin pair don't necessarily induce a unique largest set of accepting states. Instead, we exploit the fact that $\omega$-regular expressions, and consequently DRAs, are closed under complementation \cite{farwer2002omega}. The following theorem states that any $\omega$-regular property $\phi$ can be upper-bounded in an IMC by computing a lower bound on the complement property $ \neg \phi$.\\

\begin{theorem}
Let $\mathcal{I}$ be an IMC and $\mathcal{\overline{A}}$ be a DRA corresponding to $\neg \phi$, the complement of the $\omega$-regular property $\phi$. $(\mathcal{I} \otimes \mathcal{\overline{A}})_{\ell}$ is defined analogously as in Theorem 1. For any state $Q_i \in \mathcal{I}$,

\begin{align}
\mathcal{\widehat{P}}_{\mathcal{I}}(Q_i  \models \phi) =  \mathcal{\widehat{P}}_{(\mathcal{I} \otimes \mathcal{\overline{A}})_{\ell}}(\left<Q_i, s_0\right> \models \Diamond U^N_{\ell}) \; \; .
\end{align}
\end{theorem}\mbox{}

\begin{proof}
$\omega$-regular properties are closed under complementation . Therefore, for any adversary $\mathcal{\nu} \in \mathcal{\nu}_{\mathcal{I}}$ of $\mathcal{I}$, it must hold that
\begin{align*}
 \mathcal{P}_{\mathcal{I}[\nu]}(Q_i  \models \phi) & = 1 -  \mathcal{P}_{\mathcal{I}[\nu]}(Q_i  \models \neg \phi)\\
 & \leq 1 - \mathcal{P}_{\mathcal{I}[\mathcal{\nu}_{\ell}]}(Q_i  \models \neg \phi) \; \; ,
\end{align*}

\noindent where $\mathcal{\nu}_{\ell}$ is some adversary of $\mathcal{I}$ such that, for all adversary $\mathcal{\nu}$ of $\mathcal{I}$, $\mathcal{P}_{\mathcal{I}[\nu_{\ell}]}(Q_i  \models \neg \phi) \leq \mathcal{P}_{\mathcal{I}[\nu]}(Q_i  \models \neg \phi)$. By theorem 1, we have that 
\begin{align*}
\mathcal{P}_{\mathcal{I}[\mathcal{\nu}_{\ell}]}(Q_i  \models \neg \phi) = 1 - \mathcal{\widehat{P}}_{(\mathcal{I} \otimes \mathcal{\overline{A}})_{\ell}}(\left<Q_i, s_0\right> \models \Diamond U^N_{\ell}) \;\; .
\end{align*}

\noindent Therefore, the first inequality reduces to 
\begin{align*}
\mathcal{\widehat{P}}_{\mathcal{I}}(Q_i  \models \phi) = \mathcal{\widehat{P}}_{(\mathcal{I} \otimes \mathcal{\overline{A}})_{\ell}}(\left<Q_i, s_0\right> \models \Diamond U^N_{\ell}) \; \; .
\end{align*}
\end{proof}
\subsection{Search Algorithm}

\begin{algorithm}[t!]

    \SetKwInOut{Input}{Input}
    \SetKwInOut{Output}{Output}
    \Input{Interval-valued Markov Chain $\mathcal{I}$, $\omega$-regular property $\phi$.}
    \Output{Lower and upper bound probabilities of satisfying $\phi$ in $\mathcal{I}$, $\widecheck{\mathcal{P}}_{\mathcal{I}}(Q_{i} \models \phi)$ and $\widehat{\mathcal{P}}_{\mathcal{I}}(Q_{i} \models \phi)$, for all initial states $Q_{i}$.}
\mbox{}\\Construct a DRA $\mathcal{A}$ corresponding to $\phi$;  \mbox{}\\ Generate the product $\mathcal{I} \otimes \mathcal{A}$; \mbox{}\\ Find the largest set of non-accepting states $U^{N}_{\ell}$ in \mbox{}\\ $\mathcal{I} \otimes \mathcal{A}$ according to our search algorithm; \mbox{}\\ Compute $\widecheck{\mathcal{P}}_{\mathcal{I}}(Q_{i} \models \phi)$ for all $Q_{i}$ using (8) and \mbox{}\\ the reachability algorithm in  \cite{chatterjee2008model}; \mbox{}\\ \mbox{}\\
	\uIf{$|Acc| = 1$}{
	Find the largest set of non-accepting states $U^{A}_{u}$ in \mbox{}\\ $\mathcal{I} \otimes \mathcal{A}$ according to our search algorithm; \mbox{}\\ Compute $\widehat{\mathcal{P}}_{\mathcal{I}}(Q_{i} \models \phi)$ for all $Q_{i}$ using (9) and \mbox{}\\ the reachability algorithm in \cite{chatterjee2008model};
	} 
	\Else{Construct a DRA $\overline{\mathcal{A}}$ corresponding to $\overline{\phi}$, the \mbox{}\\ complement of $\phi$; \mbox{}\\  Generate the product $\mathcal{I} \otimes \overline{\mathcal{A}}$; \mbox{}\\ Find the largest set of non-accepting states $U^{N}_{\ell}$ in \mbox{}\\ $\mathcal{I} \otimes \overline{\mathcal{A}}$ according to our search algorithm; \mbox{}\\ Compute $\widehat{\mathcal{P}}_{\mathcal{I}}(Q_{i} \models \phi)$ for all $Q_{i}$ using (10) and \mbox{}\\ the reachability algorithm in \cite{chatterjee2008model};}\mbox{}\\
	
	\Return{$\widecheck{\mathcal{P}}_{\mathcal{I}}(Q_{i} \models \phi)$, $\widehat{\mathcal{P}}_{\mathcal{I}}(Q_{i} \models \phi)$}
\caption{Probability bounds computation for $\omega$-regular properties in IMCs}
\end{algorithm}

After proving the existence of the sets $U^{N}_{l}$ and $U^{A}_{u}$, we design a search algorithm for finding these sets in a given product IMC $\mathcal{I} \otimes \mathcal{A}$. Well-known techniques, such as Kosaraju's algorithm \cite{sharir1981strong}, list all \textit{strongly connected components} (SCC) in a directed graph. SCCs and BSCCs are defined similarly with the difference that the states in an SCC are permitted to transition outside of it. We seek to exploit these algorithms to detect the largest sets of accepting and non-accepting BSCCs via graph search. 

Our strategy is as follows: first, we assume that all transitions with non-zero upper bounds in $\mathcal{I} \otimes \mathcal{A}$ are effectively non-zero. The resulting product IMC induces a directed graph with a vertex for each state and an edge for all non-zero transitions. In this graph, all SCCs are enumerated. Then, for each SCC and if necessary, we remove the states that prevent it from being a BSCC and accepting (or non-accepting) if allowed by $\mathcal{I} \otimes \mathcal{A}$. Below is a detailed description of the algorithm.\\

\begin{itemize}
\item Generate a directed graph $G(V,E)$ with a vertex for each state in $\mathcal{I} \otimes \mathcal{A}$. An edge links two states $\left<Q_i, s_i\right>$ and $\left<Q_i', s_i' \right>$ if $\widehat{T}_{\left<Q_{i},s_i \right> \rightarrow \left<Q_{i'},s_{i'} \right>} > 0$,
\item Find all strongly connected components in $G$ and list them in $C$,
\item For all SCC $ C^{j} \in C$, check whether it contains a \textit{leaky} state: a state $\left<Q_i, s_i\right> \in C^{j}$ is leaky if, for some state $\left<Q_i', s_i'\right> \not \in C^{j}$, $\widecheck{T}_{\left<Q_{i},s_i \right> \rightarrow \left<Q_{i'},s_{i'} \right>} > 0$ or if $\Sigma_{\left<Q_i', s_i'\right> \in C^j} \widehat{T}_{\left<Q_{i},s_i \right> \rightarrow \left<Q_{i'},s_{i'} \right>} < 1$ (that is, $\left<Q_i, s_i\right>$ has a non-zero probability of transitioning outside of $C_j$ for all refinement of $\mathcal{I} \otimes \mathcal{A}$).
\item If a state $\left<Q_i, s_i\right> \in C^{j}$ is leaky, it cannot belong to a BSCC. Find all states $\left<Q_i', s_i'\right>$ in $C^{j}$ whose transition to a leaky state cannot be "turned off" as in the previous step. These states are leaky as well. Repeat for all leaky states in $C^{j}$.
\item In the subgraph $G^{j}$ induced by $C^{j}$, remove all edges from non-leaky to leaky states. Find all SCCs in $G^{j}$ and add them to $C$,
\item If $C^{j}$ has no leaky state, $C^{j}$ is a BSCC. For all states in $C^{j}$, check if it maps to some accepting set $F_{i}$. If not, $C^{j}$ is a non-accepting BSCC. Otherwise, we treat two different cases depending on which set of states is currently being searched for.
\item \underline{Search for $U^{N}_{l}$}:  For all such $F_{i}$'s, check whether some state in $C^{j}$ maps to the corresponding non-accepting set $E_{i}$. If this is the case for all such $F_{i}$'s, $C^{j}$ is a non-accepting BSCC. Otherwise, the unmatched $F_{i}$ states cannot belong to a non-accepting BSCC. Treat them as leaky and follow the same procedure as before for eliminating leaky states. Add the new SCCs to $C$.
\item \underline{Search for $U^{A}_{u}$}: Check whether some state in $C^{j}$ maps to $E_{1}$ (this algorithm is only valid for automata with one Rabin Pair). If not, $C^{j}$ is accepting and is $U^{A}_{u}$. Otherwise, treat the states mapping to $E_{1}$ as leaky and follow the same procedure as before for eliminating leaky states. Add the new SCCs to $C$.
\item By Lemma 2, $U^{N}_{l}$ is the union of all non-accepting BSCCs found by the algorithm.\\
\end{itemize}
It is interesting to note that the reachability problem can be solved in the original product IMC without having to explicitly construct a worst-case NASIMC or best-case ASIMC. Indeed, we know that, for any state belonging to these sets, the reachability probability is trivially 1. For all states not in these sets, a worst-case NASIMC and best-case ASIMC have the same probabilities of transitions as the original product IMC, as seen in Remark 1.

Algorithm 1 summarizes the entire procedure for bounding the satisfiability of $\omega$-regular properties in IMCs presented in this work.
\section{CASE STUDY}

We now apply the concepts developed in the previous sections to a case study. Our system of interest is an agent moving stochastically on a two-dimensional grid shown in Fig. 3. The grid is divided into 6 locations, representing 6 different states the agent can visit. We assume the system to be evolving in a discrete-time fashion: at each time instance $t$, the agent makes a transition from its current state to a new state according to some probability distribution. The latter depends only on the current state of the agent, i.e. the system satisfies the Markov property. 

\begin{figure}
\vspace{0.2cm}
\begin{center}
\includegraphics[scale=0.34]{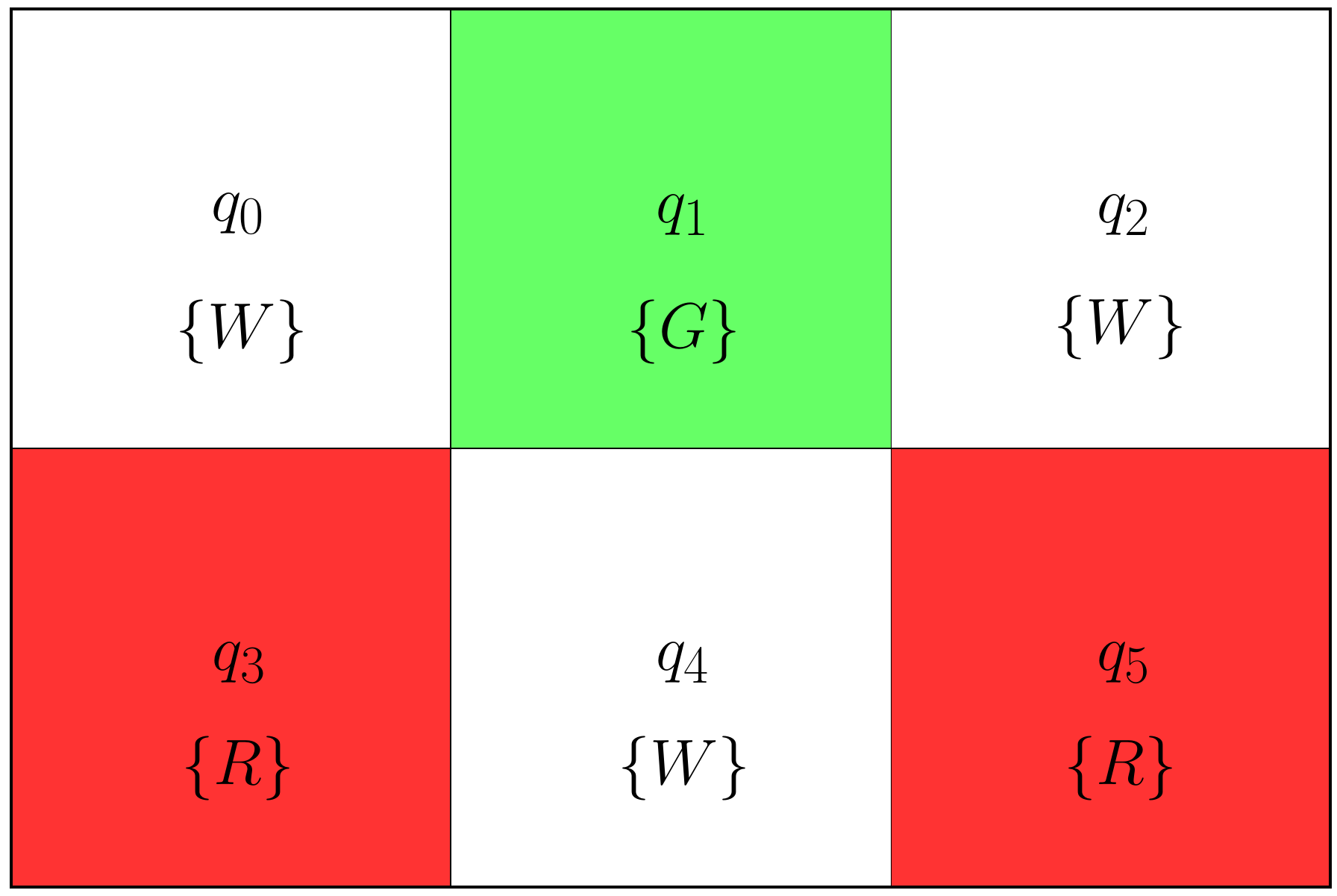}
\end{center}
\caption{A grid representation of the 6 states the agent can be in.}
\end{figure}

\begin{figure}
\begin{center}
\hspace*{-1.3cm}
\includegraphics[scale=0.50]{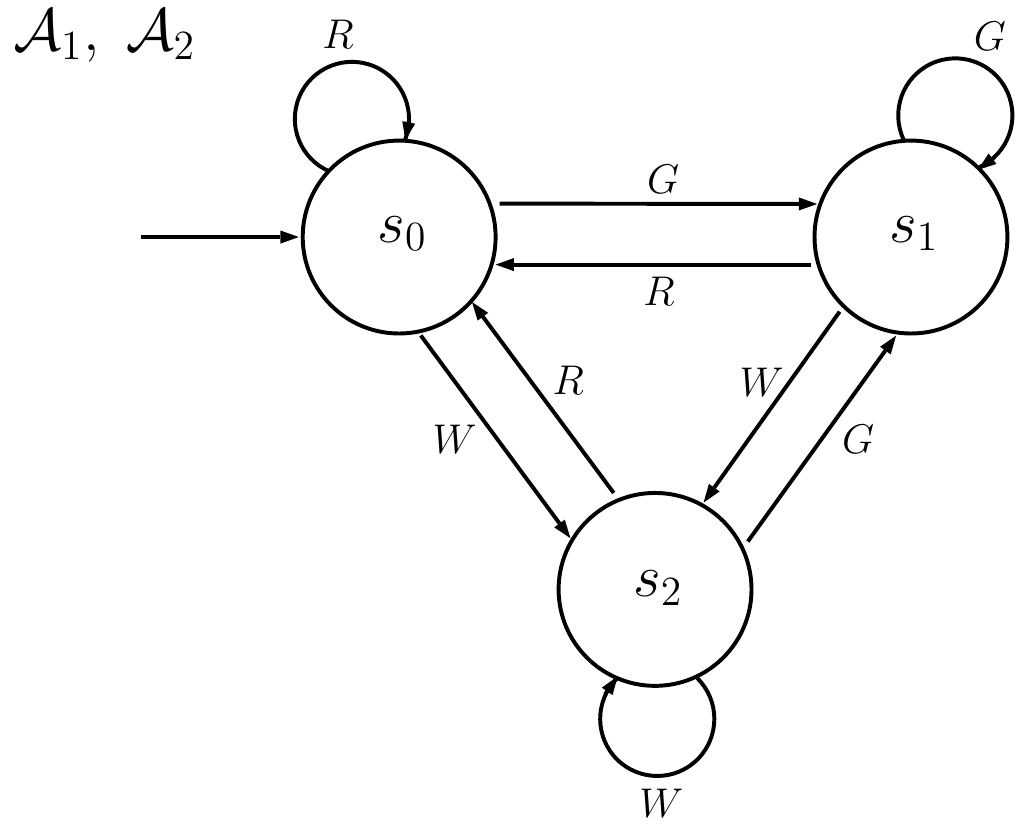}
\end{center}
\caption{A state diagram for automata $\mathcal{A}_{1}$ and $\mathcal{A}_{2}$ corresponding to properties $\phi_{1}$ and $\phi_{2}$ respectively. While their set of states and transition function are identical, they display different acceptance conditions: for $\mathcal{A}_{1}$, $Acc = \lbrace ( \{s_{0} \},  \{ s_{1} \} ) \rbrace$ whereas for $\mathcal{A}_{2}$,  $Acc = \lbrace ( \emptyset, \{ s_{1} \} ), \; ( \{ s_{0}, s_{1} \}, \{ s_{2} \} ) \rbrace$. }
\end{figure}

However, the transition probabilities are not known exactly and an IMC representation of the system is constructed. Matrices $\widehat{T}$ and $\widecheck{T}$ respectively contain the upper and lower probabilities of transition from state to state:

\begin{center}
\begin{tabular}{c || c  c  c  c c c}
$\widecheck{T}$ & $q_0$ & $q_1$ & $q_2$ & $q_3$ & $q_4$ & $q_5$\\
\hline
\hline
$q_0$ & $0$ & $0.2$ & $0$ & $0.3$ & $0.2$ & $0$ \\
$q_1$ & $0$ & $0.05$ & $0.25$ & $0$ & $0.1$ & $0$\\
$q_2$ & $0$ & $0$ & $0$ & $0$ & $1$ & $0$\\
$q_3$ & $0$ & $0$ & $0$ & $1$ & $0$ & $0$\\
$q_4$ & $0$ & $0$ & $1$ & $0$ & $0$ & $0$\\
$q_5$ & $0$ & $0.3$ & $0.2$ & $0$ & $0.2$ & $0$\\
\end{tabular}
\end{center}

\begin{center}
\begin{tabular}{c || c  c  c  c c c}
$\widehat{T}$ & $q_0$ & $q_1$ & $q_2$ & $q_3$ & $q_4$ & $q_5$ \\
\hline
\hline
$q_0$ & $0$ & $0.5$ & $0.3$ & $0.6$ & $0.5$ & $0$\\
$q_1$ & $0.2$ & $0.8$ & $0.6$ & $0.8$ & $0.7$ & $0.5$\\
$q_2$ & $0$ & $0$ & $0$ & $0$ & $1$ & $0$\\
$q_3$ & $0$ & $0$ & $0$ & $1$ & $0$ & $0$\\
$q_4$ & $0$ & $0$ & $1$ & $0$ & $0$ & $0$\\
$q_5$ & $0$ & $0.5$ & $0.5$ & $0$ & $0.3$ & $0$\\
\end{tabular}
\end{center}

\begin{table}[t!]
\vspace{0.4cm}
\centering
\begin{tabular}{| c | c | c | c | c | c | c |}
\cline{2-7}
\multicolumn{1}{c|}{} & $q_{0}$ & $q_{1}$ & $q_{2}$ & $q_{3}$ & $q_{4}$ & $q_{5}$\\
\hline 
Lower bound for $\phi_{1}$& $0$ & $0$ & $0$ & $0$ & $0$ & $0$\\
\hline
Upper bound for $\phi_{1}$& $0$ & $0$ & $0$ & $0$ & $0$ & $0$\\
\hline
Lower bound for $\phi_{2}$& $0.274$ & $0.368$ & $1$ & $0$ & $1$ & $0.684$\\
\hline
Upper bound for $\phi_{2}$& $0.7$ & $1$ & $1$ & $0$ & $1$ & $1$\\
\hline
\end{tabular}
\caption{Table compiling the computed lower and upper bounds for the satisfiability of $\phi_{1}$ and $\phi_{2}$ for all initial states.}
 \vspace{-1cm}
\end{table}

\noindent Each state is labeled as follows: $L(q_0) = L(q_2) = L(q_4) = \lbrace W \rbrace$, $L(q_1)  = \lbrace G \rbrace$ and $L(q_3) =  L(q_5) = \lbrace R \rbrace$. We aim to bound the probability of satisfying $\omega$-regular properties $\phi_{1}$ and $\phi_{2}$, represented by automata $\mathcal{A}_{1}$ and $\mathcal{A}_{2}$ in Fig. 4, from every initial state $q$. In natural language, these properties respectively translate to \\

\noindent 1) \textit{``The agent visits a green state infinitely many times while visiting a red state finitely many times."},\\

\noindent 2) \textit{``The agent shall visit a red state infinitely many times only if it visits a green state infinitely many times."}\\

\noindent Note that $\mathcal{A}_{2}$ has 2 Rabin pairs. According to Theorem 3, we thus have to construct the automata for the complement of $\phi_{2}$. $\phi_{2}$ is expressed in LTL as $\phi_{2} = \square \Diamond G \; \lor \; \Diamond \square W$ which, when complemented, becomes $\neg \phi_{2} = \Diamond \square \overline{G} \wedge  \square \Diamond \overline{W}$. Then, we convert $\neg \phi_{2}$ to a DRA with one of the existing LTL-to-$\omega$-automata translation tools \cite{komarkova2014rabinizer}. Bounds for $\phi_{1}$  and $\phi_{2}$ are computed using Algorithm 1 and are shown in Table 1.

\section{CONCLUSIONS}

We derived an efficient automaton-based technique for bounding the probability of satisfying any $\omega$-regular property in an IMC interpreted as an IMDP. We demonstrated its application through a case study. In future works, we will seek to exploit the mechanisms unveiled in this paper and apply them to Bounded-parameter Markov Decision Processes, the controllable counterparts of IMCs, e.g. to minimize or maximize the probability of occurrence of some behavior in a system.



\bibliographystyle{IEEEtran}

\bibliography{CDC2018.bib}

\begin{thebibliography}{10}
\providecommand{\url}[1]{#1}
\csname url@samestyle\endcsname
\providecommand{\newblock}{\relax}
\providecommand{\bibinfo}[2]{#2}
\providecommand{\BIBentrySTDinterwordspacing}{\spaceskip=0pt\relax}
\providecommand{\BIBentryALTinterwordstretchfactor}{4}
\providecommand{\BIBentryALTinterwordspacing}{\spaceskip=\fontdimen2\font plus
\BIBentryALTinterwordstretchfactor\fontdimen3\font minus
  \fontdimen4\font\relax}
\providecommand{\BIBforeignlanguage}[2]{{%
\expandafter\ifx\csname l@#1\endcsname\relax
\typeout{** WARNING: IEEEtran.bst: No hyphenation pattern has been}%
\typeout{** loaded for the language `#1'. Using the pattern for}%
\typeout{** the default language instead.}%
\else
\language=\csname l@#1\endcsname
\fi
#2}}
\providecommand{\BIBdecl}{\relax}
\BIBdecl

\bibitem{ding2014optimal}
X.~Ding, S.~L. Smith, C.~Belta, and D.~Rus, ``{Optimal control of Markov
  decision processes with linear temporal logic constraints},'' \emph{IEEE
  Transactions on Automatic Control}, vol.~59, no.~5, pp. 1244--1257, 2014.

\bibitem{fu2014probably}
J.~Fu and U.~Topcu, ``{Probably approximately correct MDP learning and control
  with temporal logic constraints},'' \emph{arXiv preprint arXiv:1404.7073},
  2014.

\bibitem{baier2008principles}
C.~Baier, J.-P. Katoen, and K.~G. Larsen, \emph{Principles of model
  checking}.\hskip 1em plus 0.5em minus 0.4em\relax MIT press, 2008.

\bibitem{abate2011approximate}
A.~Abate, A.~D'Innocenzo, and M.~D. Di~Benedetto, ``{Approximate abstractions
  of stochastic hybrid systems},'' \emph{IEEE Transactions on Automatic
  Control}, vol.~56, no.~11, pp. 2688--2694, 2011.

\bibitem{zamani2014symbolic}
M.~Zamani, P.~M. Esfahani, R.~Majumdar, A.~Abate, and J.~Lygeros, ``Symbolic
  control of stochastic systems via approximately bisimilar finite
  abstractions,'' \emph{IEEE Transactions on Automatic Control}, vol.~59,
  no.~12, pp. 3135--3150, 2014.

\bibitem{dutreix2018}
M.~Dutreix and S.~Coogan, ``{Efficient verification for stochastic mixed
  monotone systems},'' in \emph{International Conference on Cyber-Physical
  Systems}, 2018.

\bibitem{kozine2002interval}
I.~O. Kozine and L.~V. Utkin, ``{Interval-valued finite Markov chains},''
  \emph{Reliable computing}, vol.~8, no.~2, pp. 97--113, 2002.

\bibitem{vskulj2009discrete}
D.~{\v{S}}kulj, ``{Discrete time Markov chains with interval probabilities},''
  \emph{International journal of approximate reasoning}, vol.~50, no.~8, pp.
  1314--1329, 2009.

\bibitem{chatterjee2008model}
K.~Chatterjee, K.~Sen, and T.~Henzinger, ``{Model-checking $\omega$-regular
  properties of interval Markov chains},'' \emph{Foundations of Software
  Science and Computational Structures}, pp. 302--317, 2008.

\bibitem{lahijanian2015formal}
M.~Lahijanian, S.~B. Andersson, and C.~Belta, ``Formal verification and
  synthesis for discrete-time stochastic systems,'' \emph{IEEE Transactions on
  Automatic Control}, vol.~60, no.~8, pp. 2031--2045, 2015.

\bibitem{benedikt2013ltl}
M.~Benedikt, R.~Lenhardt, and J.~Worrell, ``{LTL model checking of interval
  Markov chains},'' in \emph{International Conference on Tools and Algorithms
  for the Construction and Analysis of Systems}.\hskip 1em plus 0.5em minus
  0.4em\relax Springer, 2013, pp. 32--46.

\bibitem{klein2006experiments}
J.~Klein and C.~Baier, ``Experiments with deterministic $\omega$-automata for
  formulas of linear temporal logic,'' \emph{Theoretical Computer Science},
  vol. 363, no.~2, pp. 182--195, 2006.

\bibitem{babiak2013effective}
T.~Babiak, F.~Blahoudek, M.~K{\v{r}}et{\'\i}nsk{\`y}, and J.~Strej{\v{c}}ek,
  ``{Effective translation of LTL to deterministic Rabin automata: Beyond the
  (F, G)-fragment},'' in \emph{International Symposium on Automated Technology
  for Verification and Analysis}.\hskip 1em plus 0.5em minus 0.4em\relax
  Springer, 2013, pp. 24--39.

\bibitem{chen2013complexity}
T.~Chen, T.~Han, and M.~Kwiatkowska, ``On the complexity of model checking
  interval-valued discrete time markov chains,'' \emph{Information Processing
  Letters}, vol. 113, no.~7, pp. 210--216, 2013.

\bibitem{farwer2002omega}
B.~Farwer, ``$\omega$-automata,'' in \emph{Automata logics, and infinite
  games}.\hskip 1em plus 0.5em minus 0.4em\relax Springer, 2002, pp. 3--21.

\bibitem{sharir1981strong}
M.~Sharir, ``{A strong-connectivity algorithm and its applications in data flow
  analysis},'' \emph{Computers \& Mathematics with Applications}, vol.~7,
  no.~1, pp. 67--72, 1981.

\bibitem{komarkova2014rabinizer}
Z.~Kom{\'a}rkov{\'a} and J.~K{\v{r}}et{\'\i}nsk{\`y}, ``{Rabinizer 3: Safraless
  translation of LTL to small deterministic automata},'' in \emph{International
  Symposium on Automated Technology for Verification and Analysis}.\hskip 1em
  plus 0.5em minus 0.4em\relax Springer, 2014, pp. 235--241.

\end{thebibliography}

\end{document}